\documentclass[reqno,centertags,draft]{amsart}
\usepackage{amsmath,amsthm,amscd,amssymb,latexsym,upref}
\date{\today} 

\newcommand{\bbN}{{\mathbb{N}}}
\newcommand{\bbR}{{\mathbb{R}}}

\newcommand{\bbZ}{{\mathbb{Z}}}
\newcommand{\bbC}{{\mathbb{C}}}

\newcommand{\cM}{{\mathcal M}}

\newcommand{\beq}{\begin{equation}}
\newcommand{\eeq}{\end{equation}}
\newcommand{\ba}{\begin{align}}
\newcommand{\ea}{\end{align}}
\newcommand{\bi}{\bibitem}
\newcommand{\no}{\notag}
\newcommand{\lb}{\label}
\newcommand{\f}{\frac}
\newcommand{\bs}{\backslash}

\newcommand{\ol}{\overline}

\newcommand{\wti}{\widetilde}

\newcommand{\st}{\,|\,}

\newcommand{\loc}{\text{\rm{loc}}}

\newcommand{\Arc}{\text{\rm{Arc}}}

\newcommand{\ess}{\text{\rm{ess}}}

\newcommand{\supp}{\text{\rm{supp}}}

\newcommand{\Om}{\Omega}
\newcommand{\om}{\omega}
\newcommand{\si}{\sigma}

\newcommand{\al}{\alpha}
\newcommand{\be}{\beta}
\newcommand{\Ga}{\Gamma}
\newcommand{\ga}{\gamma}
\newcommand{\De}{\Delta}

\newcommand{\Te}{\Theta}
\newcommand{\te}{\theta}
\newcommand{\ze}{\zeta}

\newcommand{\C}{\mathbb{C}}

\newcommand{\R}{\mathbb{R}}
\newcommand{\D}{\mathbb{D}}
\newcommand{\dD}{{\partial\hspace*{.2mm}\mathbb{D}}}
\newcommand{\Z}{\mathbb{Z}}
\newcommand{\N}{\mathbb{N}}

\newcommand{\U}{\mathbb{U}}
\newcommand{\V}{\mathbb{V}}
\newcommand{\W}{\mathbb{W}}

\newcommand{\Cm}{{\mathbb{C}^{m\times m}}}

\renewcommand{\Re}{\text{\rm Re}}
\renewcommand{\Im}{\text{\rm Im}}
\renewcommand{\ln}{\text{\rm ln}}

\newcommand{\abs}[1]{\left\lvert#1\right\rvert}
\newcommand{\norm}[1]{\left\Vert#1\right\Vert}

\newcommand{\ltm}[1]{{\ell^2(#1)^m}}

\newcommand{\ltmm}[1]{{\ell^2(#1)^{m\times m}}}

\newcommand{\Ltm}[1]{{L^2(\dD;d\Om_{#1}(\cdot,k_0))}}

\allowdisplaybreaks
\numberwithin{equation}{section}

\theoremstyle{plain}
\newtheorem{theorem}{Theorem}[section]

\newtheorem{hypothesis}[theorem]{Hypothesis}
\theoremstyle{definition}
\newtheorem{definition}[theorem]{Definition}
\newtheorem{remark}[theorem]{Remark}

\begin{document}

\title[Trace Formulas and a Borg-Type Theorem for CMV Operators]
{Trace Formulas and a Borg-type Theorem for CMV Operators with
Matrix-valued Coefficients}
\author[M.\ Zinchenko]{Maxim Zinchenko}
\address{Department of Mathematics,
California Institute of Technology, Pasadena, CA 91125, USA}
\email{maxim@caltech.edu}
\urladdr{http://www.math.caltech.edu/\textasciitilde maxim}



\begin{abstract}
We prove a general Borg-type inverse spectral result for a
reflectionless unitary CMV operator (CMV for Cantero, Moral, and
Vel\'azquez \cite{CMV03}) associated with matrix-valued Verblunsky
coefficients. More precisely, we find an explicit formula for the
Verblunsky coefficients of a reflectionless CMV matrix whose spectrum
consists of a connected arc on the unit circle. This extends a recent
result \cite{GZ06} for CMV operators with scalar-valued coefficients.

In the course of deriving the Borg-type result we also use
exponential Herglotz representations of Caratheodory matrix-valued
functions to prove an infinite sequence of trace formulas connected
with CMV operators.
\end{abstract}


\maketitle

\section{Introduction}\label{s1}

The aim of this paper is to prove a Borg-type uniqueness theorem for
a special class of unitary doubly infinite block-five-diagonal
matrices, that is, doubly infinite CMV matrices associated with
matrix-valued Verblunsky coefficients. The actual history of CMV
matrices (with scalar coefficients) is quite interesting: The
corresponding unitary semi-infinite five-diagonal matrices were
first introduced in 1991 by Bunse--Gerstner and Elsner \cite{BGE91},
and subsequently discussed in detail by Watkins \cite{Wa93} in 1993
(cf.\ also the recent discussion in Simon \cite{Si06}). They were
subsequently rediscovered by Cantero, Moral, and Vel\'azquez (CMV)
in \cite{CMV03}. In \cite[Sects.\ 4.5, 10.5]{Si04}, Simon introduced
the corresponding notion of unitary doubly infinite five-diagonal
matrices and coined the term ``extended'' CMV matrices. For
simplicity, we will just speak of CMV operators whether or not they
are half-lattice or full-lattice operators. We also note that in a
context different from orthogonal polynomials on the unit circle,
Bourget, Howland, and Joye \cite{BHJ03} introduced a family of
doubly infinite matrices with three sets of parameters which, for
special choices of the parameters, reduces to two-sided CMV matrices
on $\bbZ$. Moreover, it is possible to connect unitary block Jacobi
matrices to the trigonometric moment problem (and hence to CMV
matrices) as discussed by Berezansky and Dudkin \cite{BD05},
\cite{BD06}. CMV operators with matrix-valued coefficients were
recently discussed in \cite{CGZ07}, \cite{DKS08}, \cite{DPS08},
\cite{Si06}.

The relevance of CMV operators, more precisely, half-lattice CMV
operators is derived from its intimate relationship with the
trigonometric moment problem and hence with finite measures on the
unit circle $\dD$. For a detailed account of the relationship of
half-lattice CMV operators with orthogonal polynomials on the unit
circle we refer to the monumental two-volume treatise by Simon
\cite{Si04} (see also \cite{Si05}, \cite{Si06}) and the exhaustive
bibliography therein. For classical results on orthogonal
polynomials on the unit circle we refer, for instance, to
\cite{Ak65}, \cite{Ge46}--\cite{Ge61}, \cite{Kr45},
\cite{Sz20}--\cite{Sz78}, \cite{Ve35}, \cite{Ve36}. More recent
references relevant to the spectral theoretic content of this paper
are \cite{CGZ07}, \cite{GJ96}--\cite{GT94}, \cite{GZ06},
\cite{GZ06a}, \cite{GN01}, \cite{Lu04}, \cite{PY04}, \cite{Si04a}.
Moreover, the full-lattice CMV operators are closely related to an
important, and only recently intensively studied, completely
integrable nonabelian version of the defocusing nonlinear
Schr\"odinger equation (continuous in time but discrete in space), a
special case of the Ablowitz--Ladik system. Relevant references in
this context are, for instance, \cite{AL75}--\cite{APT04},
\cite{GGH05}, \cite{GH05}--\cite{GHMT08b}, \cite{KN07}, \cite{Li05},
\cite{MEKL95}--\cite{Ne06}, \cite{Sc89}, \cite{Ve99}, and the
literature cited therein. We emphasize that the case of
matrix-valued Verblunsky coefficients is considerably less studied
than the case of scalar coefficients, but we refer to \cite{CGZ07},
\cite{DKS08}, \cite{DPS08}, \cite{Si06}.

From the outset, Borg-type theorems are inverse spectral theory
assertions which typically prescribe a connected interval (or arc)
as the spectrum of a self-adjoint (or unitary) differential or
difference operator, and under a reflectionless condition imposed on
the operator (one may think of a periodicity condition on the
(potential) coefficients of the differential or difference operator)
infers an explicit form of the coefficients of the operator in
question. Typically, the form of the coefficients determined in this
context is fairly simple and usually given by constants or functions
of exponential type.

Next, we briefly describe the history of Borg-type theorems relevant
to this paper. In 1946, Borg \cite{Bo46} proved, among a variety of
other inverse spectral theorems, the following result for
one-dimensional Schr\"odinger operators. (Throughout this paper we
denote by $\sigma(\cdot)$ and $\sigma_{\ess}(\cdot)$ the spectrum
and essential spectrum of a densely defined closed linear operator
in a complex separable Hilbert space.)

\begin{theorem}[\cite{Bo46}] \lb{t1.1}
Let $q\in L^1_{\loc} (\bbR)$ be real-valued and periodic,
$H=-\f{d^2}{dx^2}+q$ be the associated self-adjoint Schr\"odinger
operator in $L^2(\bbR)$, and suppose that
\begin{equation}
\sigma(H)=[e_0,\infty) \, \text{ for some $e_0\in\bbR$.}
\end{equation}
Then $q$ is of the form,
\begin{equation}
q(x)=e_0 \, \text{ for a.e.\ $x\in\bbR$}. \lb{1.5}
\end{equation}
\end{theorem}

Traditionally, uniqueness results such as
Theorem\ \ref{t1.1} are called Borg-type theorems. However, this
terminology is not uniquely adopted  and hence a bit unfortunate.
Indeed, inverse spectral results on finite intervals in which the
coefficient(s) in the underlying differential or difference expression are
recovered from two spectra, were also pioneered by Borg in his celebrated
paper \cite{Bo46}, and hence are also coined Borg-type theorems in the
literature, see, for instance, \cite{Ma94}, \cite{Ma99a}.

A closer examination of the proof of Theorem\ \ref{t1.1} in
\cite{CGHL00} shows that periodicity of $q$ is not the point for the
uniqueness result \eqref{1.5}. The key ingredient (besides
$\sigma(H)=[e_0,\infty)$ and $q$ real-valued) is the fact that
\begin{equation}
\text{for all $x\in \bbR$, } \, \xi(\lambda,x)=1/2 \, \text{ for a.e.\
$\lambda\in\sigma_{\ess}(H)$.} \lb{1.6}
\end{equation}
Here $\xi(\lambda,x)$, the argument of the boundary value
$g(\lambda+i0,x)$ of the diagonal Green's function of $H$ on the real
axis (where $g(z,x)=(H-zI)^{-1}(x,x)$, $z\in\bbC\backslash\sigma(H)$,
$x\in\R$), is defined by
\begin{equation}
\xi(\lambda,x)=\pi^{-1}\lim_{\varepsilon\downarrow 0}
\Im(\ln(g(\lambda+i\varepsilon,x))) \, \text{
for a.e.\ $\lambda\in\bbR$ and all $x\in\R$}. \lb{1.7}
\end{equation}

Real-valued periodic potentials are known to satisfy \eqref{1.6},
but so do certain classes of real-valued quasi-periodic and
almost-periodic potentials $q$. In particular, the class of
real-valued algebro-geometric finite-gap KdV potentials $q$ (a
subclass of the set of real-valued quasi-periodic  potentials) is a
prime example satisfying \eqref{1.6} without necessarily being
periodic. Traditionally, potentials $q$ satisfying \eqref{1.6} are
called \textit{reflectionless} (see \cite{BGMS05}, \cite{CG02},
\cite{CGHL00}, \cite{CGR05}, and the references therein).

The extension of Borg's Theorem\ \ref{t1.1} to periodic
matrix-valued Schr\"odinger operators was first proved by D\'epres
\cite{De95}. A new strategy of the proof based on exponential
Herglotz representations and a trace formula (cf.\ \cite{GS96}) for
such potentials, as well as the extension to reflectionless
matrix-valued potentials, was obtained in \cite{CGHL00}.

The direct analog of Borg's Theorem\ \ref{t1.1} for periodic Jacobi
operators was proved by Flaschka \cite{Fl75} in 1975.

\begin{theorem}  [\cite{Fl75}] \lb{t1.2}
Suppose $a=\{a_k\}_{k\in\bbZ}$ and $b=\{b_k\}_{k\in\bbZ}$
are  periodic real-valued sequences in
$\ell^\infty(\bbZ)$ with the same period and $a_k>0$, $k\in\bbZ$. Let
$H=aS^+ +a^-S^- +b$ be the associated self-adjoint  Jacobi operator
on $\ell^2(\bbZ)$ and suppose that
\begin{equation}
\sigma(H)=[E_-,E_+] \, \text{ for some $E_-<E_+$.}
\end{equation}
Then $a=\{a_k\}_{k\in\bbZ}$ and $b=\{b_k\}_{k\in\bbZ}$ are of the form,
\begin{equation}
a_k=(E_+-E_-)/4 , \quad  b_k=(E_-+E_+)/2, \quad k\in\bbZ.  \lb{1.8}
\end{equation}
\end{theorem}

Here $S^{\pm}$ denote the shift operators $S^{\pm}f=f^\pm=f(\cdot\pm
1)$, $f\in\ell^\infty(\bbZ)$.

The extension of Theorem\ \ref{t1.2} to reflectionless scalar Jacobi
operators is due to Teschl \cite[Corollary\ 6.3]{Te98} (see also
\cite[Corollary\ 8.6]{Te00}). The extension of Theorem\ \ref{t1.2}
to matrix-valued reflectionless Jacobi operators (and a
corresponding result for Dirac-type difference operators) has
recently been obtained in \cite{CGR05}.

The first analog of Borg-type theorem for CMV operators with
periodic scalar coefficients and spectrum filling out the whole unit
circle was obtained by Simon in \cite[Sect.\ 11.14]{Si04}. It was
recently extended in \cite{GZ06} to reflectionless CMV operators
with scalar Verblunsky coefficients and spectra given by a
connected arc on the unit circle:

\begin{theorem}[\cite{GZ06}] \lb{t1.3}
Let $\alpha=\{\alpha_k\}_{k\in\bbZ}\subset\D$ denote Verblunsky
coefficients associated with a reflectionless CMV operator $\U$
$($cf. \eqref{2.1}--\eqref{2.3}$)$ on $\ell^2(\bbZ)$. Suppose that
the spectrum of $\U$ consists of a connected arc on $\dD$,
\begin{equation}
\si(\U)=\Arc\big(\big[e^{i\theta_0},e^{i\theta_1}\big]\big), \quad
\theta_0 \in [0,2\pi), \; \theta_0<\theta_1\leq\theta_0+2\pi.
\end{equation}
Then $\alpha=\{\alpha_k\}_{k\in\bbZ}$ is of the form,
\begin{equation}
\alpha_k=\alpha_0 g^k, \quad k\in\bbZ,
\end{equation}
where
\begin{equation}
g=-\exp(i(\theta_0+\theta_1)/2) \, \text{ and } \,
|\alpha_0|=\cos((\theta_1-\theta_0)/4).
\end{equation}
\end{theorem}

The main goal of this paper is to extend Theorem \ref{t1.3} to CMV
operators with matrix-valued Verblunsky coefficients, introduced in
\cite{CGZ07}. Our study will be undertaken in the spirit of
\cite{BGMS05}, \cite{CG02}, \cite{CGHL00}, \cite{CGR05}, where
Borg-type theorems were proven for matrix-valued Schr\"odinger and
Dirac-type operators on $\bbR$ and similarly for matrix-valued
Jacobi operators on $\bbZ$.

In Section \ref{s2} we prove an infinite sequence of trace formulas
connected with CMV operators using Weyl--Titchmarsh $m$-functions
(and their exponential Herglotz representations). The notion of
reflectionless CMV operators is introduced in Section \ref{s3} and a
variety of necessary conditions (many of them also sufficient) for a
CMV operator to be reflectionless are established. In Section
\ref{s4} we prove our main new result, a Borg-type theorem for
reflectionless CMV operators with matrix-valued Verblunsky
coefficients whose spectrum consists of a connected arc on the unit
circle $\dD$.

\section{Trace Formulas} \label{s2}

In this section we introduce CMV operators with matrix-valued
Verblunsky coefficients, review some basic facts on the
Weyl--Titchmarsh theory associated with these operators, and derive
an infinite sequence of trace formulas. We freely use the notation
established in Appendix \ref{sA}.

Let $\ltm{\Z}=\ell^2(\Z)\otimes\C^m$ be the Hilbert space of
sequences of $m$-dimensional complex-valued vectors with scalar
product given by
\begin{align} \lb{2.1a}
(\phi,\psi)_{\ltm{\Z}} = \sum_{k=-\infty}^\infty\sum_{j=1}^m
\ol{(\phi(k))_j}(\psi(k))_j, \quad \phi,\psi\in\ltm{\Z},
\end{align}
where we used the following notation for elements of $\ltm{\Z}$
\begin{align} \lb{2.1b}
\phi=\{\phi(k)\}_{k\in\Z}=
\begin{pmatrix}
\vdots\\\phi(-1)\\\phi(0)\\\phi(1)\\\vdots
\end{pmatrix}\in\ltm{\Z}, \quad
\phi(k)=
\begin{pmatrix}
(\phi(k))_1\\(\phi(k))_2\\\vdots\\(\phi(k))_m
\end{pmatrix}\in\C^m, \;  k\in\Z.
\end{align}
A straightforward modification of the above definitions also yields
the Hilbert space $\ltm{J}$ for any $J\subset\Z$. For simplicity, we
will abbreviate the $m\times m$ identity matrix by $I_m$ and the
identity operator on $\ltm{J}$, $J\subseteq\Z$, by $I$ without
separately indicating its dependence on $J$ or $m$.

We start by introducing our basic assumption:

\begin{hypothesis} \lb{h2.1}
Let $m\in\N$ and assume $\al=\{\al_k\}_{k\in\Z}$ is a sequence of $m
\times m$ matrices with complex entries and such that
\begin{equation} \lb{2.1}
\norm{\al_k}_{\Cm} < 1, \quad k\in\Z.
\end{equation}
\end{hypothesis}

Given a sequence $\al$ satisfying \eqref{2.1}, we define two
sequences of positive self-adjoint $m\times m$ matrices
$\{\rho_k\}_{k\in\bbZ}$ and $\{\wti\rho_k\}_{k\in\bbZ}$ by
\begin{align}
\rho_k = (I_m-\al_k^*\al_k)^{1/2}, \;\; \wti\rho_k &=
(I_m-\al_k\al_k^*)^{1/2}, \quad k\in\bbZ. \lb{2.2}
\end{align}
Then \eqref{2.2} implies that $\rho_k$ and $\wti\rho_k$ are
invertible matrices for all $k\in\Z$, and using elementary power
series expansions one verifies the following identities
\begin{align}
&\wti\rho_k^{\pm1}\al_k = \al_k\rho_k^{\pm1} \,\text{ and }\,
\al_k^*\wti\rho_k^{\pm1} = \rho_k^{\pm1}\al_k^*, \quad k\in\Z.
\lb{2.2a}
\end{align}

According to Simon \cite{Si04}, we call $\al_k$ the Verblunsky
coefficients in honor of Verblunsky's pioneering work in the theory
of orthogonal polynomials on the unit circle \cite{Ve35},
\cite{Ve36}.

Next, we introduce a sequence of $2\times 2$ block unitary matrices
$\Te_k$ with $m\times m$ matrix coefficients by
\begin{equation} \lb{2.2b}
\Te_k = \begin{pmatrix} -\al_k & \wti\rho_k \\ \rho_k & \al_k^*
\end{pmatrix},
\quad k \in \Z,
\end{equation}
and two unitary operators $\V$ and $\W$ on $\ltm{\Z}$ by their
matrix representations in the standard basis of $\ltm{\Z}$ by
\begin{align} \lb{2.2c}
\V &= \begin{pmatrix} \ddots & & &
\raisebox{-3mm}[0mm][0mm]{\hspace*{-5mm}\Huge $0$}  \\ & \Te_{2k-2}
& & \\ & & \Te_{2k} & & \\ &
\raisebox{0mm}[0mm][0mm]{\hspace*{-10mm}\Huge $0$} & & \ddots
\end{pmatrix}, \quad
\W = \begin{pmatrix} \ddots & & &
\raisebox{-3mm}[0mm][0mm]{\hspace*{-5mm}\Huge $0$}
\\ & \Te_{2k-1} &  &  \\ &  & \Te_{2k+1} &  & \\ &
\raisebox{0mm}[0mm][0mm]{\hspace*{-10mm}\Huge $0$} & & \ddots
\end{pmatrix},
\end{align}
where
\begin{align}
\begin{pmatrix}
\V_{2k-1,2k-1} & \V_{2k-1,2k} \\ \V_{2k,2k-1}   & \V_{2k,2k}
\end{pmatrix} =  \Te_{2k},
\quad
\begin{pmatrix}
\W_{2k,2k} & \W_{2k,2k+1} \\ \W_{2k+1,2k}  & \W_{2k+1,2k+1}
\end{pmatrix} =  \Te_{2k+1},
\quad k\in\Z. \lb{2.2d}
\end{align}

Finally, we define the main object of our investigation, namely the
unitary operator $\U$ on $\ltm{\Z}$ as the product of the unitary
operators $\V$ and $\W$ by
\begin{equation} \lb{2.2e}
\U = \V\W,
\end{equation}
or in matrix form in the standard basis of $\ltm{\Z}$, by
\begin{align}
\U = \begin{pmatrix} \ddots &&\hspace*{-8mm}\ddots
&\hspace*{-10mm}\ddots &\hspace*{-12mm}\ddots &\hspace*{-14mm}\ddots
&&& \raisebox{-3mm}[0mm][0mm]{\hspace*{-6mm}{\Huge $0$}}
\\
&0& -\al_{0}\rho_{-1} & -\al_{0}\al_{-1}^* & -\wti\rho_{0}\al_{1} &
\wti\rho_{0}\wti\rho_{1}
\\
&& \rho_{0}\rho_{-1} &\rho_{0}\al_{-1}^* & -\al_{0}^*\al_{1} &
\al_{0}^*\wti\rho_{1} & 0
\\
&&&0& -\al_{2}\rho_{1} & -\al_{2}\al_{1}^* & -\wti\rho_{2}\al_{3} &
\wti\rho_{2}\wti\rho_{3}
\\
&&\raisebox{-4mm}[0mm][0mm]{\hspace*{-6mm}{\Huge $0$}} &&
\rho_{2}\rho_{1} & \rho_{2}\al_{1}^* & -\al_{2}^*\al_{3} &
\al_{2}^*\wti\rho_{3}&0
\\
&&&&&\hspace*{-14mm}\ddots &\hspace*{-14mm}\ddots
&\hspace*{-14mm}\ddots &\hspace*{-8mm}\ddots &\ddots
\end{pmatrix}. \lb{2.3}
\end{align}
Here terms of the form $-\al_{2k}\al_{2k-1}^*$ and
$-\al_{2k}^*\al_{2k+1}$, $k\in\Z$, represent the diagonal entries
$\U_{2k-1,2k-1}$ and $\U_{2k,2k}$ of the infinite matrix $\U$ in
\eqref{2.3}, respectively. We continue to call the operator $\U$ on
$\ltm{\Z}$ the CMV operator since \eqref{2.2b}--\eqref{2.3} in the
context of the scalar-valued semi-infinite (i.e., half-lattice) case
were obtained by Cantero, Moral, and Vel\'azquez in \cite{CMV03} in
2003, but we refer to the discussion in the introduction about the
involved history of these operators.

Next, following \cite{CGZ07} we recall the definition of the
$2m\times2m$ matrix-valued Weyl--Titchmarsh function
$\cM(\cdot,k_0)$, $k_0\in\Z$, associated with $\U$,
\begin{align}
\cM(z,k_0) &= \begin{pmatrix} M_{0,0}(z,k_0) & M_{0,1}(z,k_0) \\
M_{1,0}(z,k_0) & M_{1,1}(z,k_0) \end{pmatrix} \no
\\ &=
\begin{pmatrix}
\De_{k_0-1}(\U+zI)(\U-zI)^{-1}\De_{k_0-1}
&\De_{k_0-1}(\U+zI)(\U-zI)^{-1}\De_{k_0}
\\
\De_{k_0}(\U+zI)(\U-zI)^{-1}\De_{k_0-1} &
\De_{k_0}(\U+zI)(\U-zI)^{-1}\De_{k_0}
\end{pmatrix} \no
\\ &=
\oint_\dD d\Omega(\ze,k_0)\, \frac{\ze+z}{\ze-z}, \quad
z\in\bbC\backslash\dD. \lb{2.3a}
\end{align}
Here $\De_{k_0}$ denote the orthogonal projections onto the
$m$-dimensional subspaces $\ltm{\{k_0\}}$, $k_0\in\Z$. The
nonnegative $2m\times2m$ matrix-valued measure $d\Om(\cdot,k_0)$,
$k_0\in\Z$, is given by
\begin{align}
d\Omega(\ze,k_0) &= d
\begin{pmatrix}
\Omega_{0,0}(\ze,k_0) & \Omega_{0,1}(\ze,k_0)
\\
\Omega_{1,0}(\ze,k_0) & \Omega_{1,1}(\ze,k_0)
\end{pmatrix} \no
\\ &= d
\begin{pmatrix}
\De_{k_0-1}E_{\U}(\ze)\De_{k_0-1} & \De_{k_0-1}E_{\U}(\ze)\De_{k_0}
\\
\De_{k_0}E_{\U}(\ze)\De_{k_0-1} & \De_{k_0}E_{\U}(\ze)\De_{k_0}
\end{pmatrix}, \quad \ze \in\dD, \lb{2.3b}
\end{align}
where $E_{\U}(\cdot)$ denotes the family of operator-valued spectral
projections of the unitary CMV operator $\U$ on $\ltm{\bbZ}$,
\begin{equation}
\U=\oint_\dD dE_{\U}(\ze)\,\ze.
\end{equation}
It follows from Theorem \ref{tA.2} that $\cM(\cdot,k_0)|_{\D}$ is a
Caratheodory matrix and the measure $d\Om(\cdot,k_0)$ can be also
obtained from $\cM(\cdot,k_0)$ via \eqref{A.4}.

The Weyl--Titchmarsh function $\cM(\cdot,k_0)$ is a fundamental
object in the spectral theory of CMV operators. It encodes all the
spectral information of the corresponding operator $\U$ which can be
illustrated by the following result.

\begin{theorem}[\cite{CGZ07}] \lb{t2.2}
The full-lattice CMV operator $\U$ is unitarily equivalent to the
operator of multiplication by $\zeta$ on $\Ltm{}$ for any
$k_0\in\Z$. In particular,
\begin{align}
& \si(\U) = \supp \, (d\Om(\cdot,k_0)), \quad k_0\in\Z,
\end{align}
where $\si(\U)$ denotes the spectrum of $\U$.
\end{theorem}

We refer to Section 3 of the recent monograph \cite{CGZ07} for a
detailed discussion of this and other relations between
Weyl--Titchmatsh function $\cM(\cdot,k_0)$ and the associated CMV
operator $\U$.

Next, we note that
\begin{align}
M_{0,0}(\cdot,k_0+1) = M_{1,1}(\cdot,k_0), \quad k_0\in\Z \lb{2.3c}
\end{align}
and
\begin{align}
M_{1,1}(z,k_0) &= \De_{k_0}(\U+zI)(\U-zI)^{-1}\De_{k_0} \no
\\ & = \oint_\dD d\Omega_{1,1}(\ze,k_0) \,
\frac{\ze+z}{\ze-z}, \quad z\in\bbC\backslash\dD,\; k_0\in\Z,
\lb{2.3d}
\end{align}
where
\begin{equation}
d\Omega_{1,1}(\ze,k_0)=d\De_{k_0}E_{\U}(\ze)\De_{k_0}, \quad
\ze\in\dD. \lb{2.3e}
\end{equation}
Thus, $M_{0,0}|_{\D}$ and $M_{1,1}|_{\D}$ are $m\times m$
Caratheodory matrices. Moreover, by \eqref{2.3d} one infers that
\begin{equation}
M_{1,1}(0,k_0)=I_m, \quad k_0\in\Z. \lb{2.3f}
\end{equation}
This implies that for any nonzero vector $x_0\in\C^m$ the
scalar-valued Caratheodory function $m_{x_0}(z,k_0)=\big(x_0,
M_{1,1}(z,k_0)x_0\big)_{\C^m}$, $z\in\D$, is not identically zero,
and hence, $\Re(m_{x_0}(z,k_0))>0$ for all $z\in\D$. Thus,
\begin{equation}
\Re(M_{1,1}(z,k_0))>0, \quad z\in\D, \; k_0\in\Z. \lb{2.3g}
\end{equation}
It is also often beneficial to introduce the Schur matrix
$\Phi_{1,1}(\cdot,k_0)$ associated to $M_{1,1}(\cdot,k_0)$ via
\begin{align}
\begin{split}
\Phi_{1,1}(z,k_0) &= [M_{1,1}(z,k_0)-I_m][M_{1,1}(z,k_0)+I_m]^{-1},
\\
M_{1,1}(z,k_0) &=
[I_m+\Phi_{1,1}(z,k_0)][I_m-\Phi_{1,1}(z,k_0)]^{-1}, \quad z\in\D.
\end{split} \lb{2.3h}
\end{align}

In analogy to the exponential representation of invertible
matrix-valued Herglotz functions (i.e., matrix-valued functions
analytic in the open complex upper half-plane $\bbC_+$ with
nonnegative imaginary part on $\bbC_+$ and invertible on $\bbC_+$,
cf.\ \cite{Ca76}, \cite{GT00}) one obtains the following result.

\begin{theorem} \lb{t2.3}
Let $F$ be an $m\times m$ Caratheodory matrix with $F(z)$ invertible
for all $z\in\D$. Then $-i\ln(iF)$ is an $m\times m$ Caratheodory
matrix and $F$ has the exponential Herglotz representation,
\begin{align}
\begin{split}
& -i\ln(iF(z))=iD + \oint_{\dD} d\mu_0(\zeta) \,\Upsilon (\zeta) \,
\f{\zeta+z}{\zeta-z},
\quad z\in\D,
\\
& \; D=-\Re(\ln(F(0))), \quad 0 \leq \Upsilon (\zeta)\leq \pi I_m
\,\text{ for $\mu_0$-a.e.\ $\zeta\in\dD$}, \lb{2.4}
\end{split}
\end{align}
where $d\mu_0$ is the normalized Lebesgue measure on $\dD$ $($cf.
\eqref{A.5b}$)$. The $m\times m$ matrix-valued function $\Upsilon$
can be reconstructed from $F$ by
\begin{align}
\Upsilon (\zeta) &= \lim_{r\uparrow 1}\Re[-i\ln(iF(r\zeta))] \no
\\
& =(\pi/2)I_m+\lim_{r\uparrow 1}\Im[\ln(F(r\zeta))] \,\text{ for
$\mu_0$-a.e.\ $\zeta\in\dD$.}
\end{align}
\end{theorem}

By Theorem \ref{t2.3} and \eqref{2.3f}, \eqref{2.3g}, one then
obtains the exponential Herglotz representation for
$M_{1,1}(\cdot,k_0)$, $k_0\in\bbZ$,
\begin{align}
\begin{split}
& -i\ln[iM_{1,1}(z,k_0)]=\oint_{\dD} d\mu_0(\zeta)\,
\Upsilon_{1,1}(\zeta,k_0) \f{\zeta+z}{\zeta-z}, \quad z\in\D,
\\
& \;\, 0\leq \Upsilon_{1,1}(\zeta,k_0) \leq \pi I_m \,\text{ for
$\mu_0$-a.e.\ $\zeta\in\dD$}.
\end{split} \lb{2.5}
\end{align}
For our present purpose it is more convenient to rewrite \eqref{2.5}
in the form
\begin{align}
\begin{split}
& \ln[M_{1,1}(z,k_0)]=i \oint_{\dD} d\mu_0(\zeta)\,
\Xi_{1,1}(\zeta,k_0) \f{\zeta+z}{\zeta-z}, \quad z\in\D,
\\
&\!  -(\pi/2)I_m \leq \Xi_{1,1}(\zeta,k_0) \leq (\pi/2)I_m \,\text{
for $\mu_0$-a.e.\ $\zeta\in\dD$},
\end{split}\lb{2.7}
\end{align}
where
\begin{align}
\Xi_{1,1}(\zeta,k_0) &= \lim_{r\uparrow 1}
\Im[\ln(M_{1,1}(r\zeta,k_0))] \no
\\
&= \Upsilon_{1,1}(\zeta,k_0)-(\pi/2)I_m  \,\text{ for $\mu_0$-a.e.\
$\zeta\in\dD$}. \lb{2.8}
\end{align}
We note that $M_{1,1}(0,k_0)=I_m$ also implies
\begin{equation}
\oint_{\dD} d\mu_0(\zeta)\, \Xi_{1,1}(\zeta,k_0)=0, \quad
k_0\in\bbZ. \lb{2.9}
\end{equation}

To derive trace formulas for $\U$ we now expand $M_{1,1}(z,k_0)$
near $z=0$ into a norm convergent series with matrix-valued
coefficients. It follows from \eqref{2.3d} that
\begin{align}
M_{1,1}(z,k_0) &= \De_{k_0}(\U+zI)(\U-zI)^{-1}\De_{k_0} =
I_m+2z\De_{k_0}\U^*(I-z\U^*)^{-1}\De_{k_0} \no
\\
&= I_m+\sum_{j=1}^\infty M_j(\U,k_0)z^j, \quad z\in\D, \lb{2.10}
\end{align}
where
\begin{equation}
M_j(\U,k_0) = 2\De_{k_0}(\U^*)^j\De_{k_0}, \quad j\in\N,\;
k_0\in\bbZ.
\end{equation}
Explicitly, using \eqref{2.3}, one computes for $k_0\in\Z$,
\begin{align}
&M_1(\U,k_0) = -2
\begin{cases}
\al_{k_0}\al_{k_0+1}^*, & k_0 \text{ odd,}
\\
\al_{k_0+1}^*\al_{k_0}, & k_0 \text{ even,}
\end{cases} \lb{2.11}
\\
&M_2(\U,k_0) \no
\\
&\quad = 2
\begin{cases}
(\al_{k_0}\al_{k_0+1}^*)^2
-\al_{k_0}\rho_{k_0+1}\al_{k_0+2}^*\wti\rho_{k_0+1}
-\wti\rho_{k_0}\al_{k_0-1}\rho_{k_0}\al_{k_0+1}^*, & k_0 \text{
odd,}
\\
(\al_{k_0+1}^*\al_{k_0})^2
-\rho_{k_0+1}\al_{k_0+2}^*\wti\rho_{k_0+1}\al_{k_0}
-\al_{k_0+1}^*\wti\rho_{k_0}\al_{k_0-1}\rho_{k_0}, & k_0 \text{
even.}
\end{cases} \lb{2.12}
\end{align}

Next, we note that the Taylor expansion \eqref{2.10} implies the
norm convergent expansion
\begin{equation}
\ln(M_{1,1}(z,k_0))=\sum_{j=1}^\infty L_j(\U,k_0)z^j, \, \text{
$|z|$ sufficiently small}, \; k_0\in\bbZ, \lb{2.16}
\end{equation}
where the matrix-valued coefficients $L_j(\U,k_0)$ can be expressed
in terms of the coefficients $M_j(\U,k_0)$, $j\in\N$, $k_0\in\Z$,
\begin{align}
L_1(\U,k_0) & = M_1(\U,k_0), \lb{2.17}
\\
L_2(\U,k_0) & = M_2(\U,k_0) - \f12 M_1(\U,k_0)^2, \lb{2.17a}
\\
L_3(\U,k_0) & = M_3(\U,k_0) -
\f12\big[M_1(\U,k_0)M_2(\U,k_0)+M_2(\U,k_0)M_1(\U,k_0)\big]
\\
&\quad + \f13 M_1(\U,k_0)^3, \;\text{ etc.}
\end{align}

\begin{theorem} \lb{t2.4}
Assume Hypothesis \ref{h2.1}. Then the following trace formulas
associated with the CMV operator $\U$ hold,
\begin{equation}
L_j(U,k_0)=2i\oint_{\dD} d\mu_0(\zeta)\, \Xi_{1,1}(\zeta,k_0)\,{\ol
\zeta}^j, \quad j\in\bbN, \; k_0\in\Z. \lb{2.18}
\end{equation}
In particular,
\begin{align}
L_1(\U,k_0) &= -2
\begin{cases}
\al_{k_0}\al_{k_0+1}^*, & k_0 \text{ odd,}
\\
\al_{k_0+1}^*\al_{k_0}, & k_0 \text{ even,}
\end{cases} \no
\\
&= 2i\oint_{\dD} d\mu_0(\zeta)\, \Xi_{1,1}(\zeta,k_0)\,{\ol \zeta},
\lb{2.19}
\\
L_2(\U,k_0) &= -2
\begin{cases}
\al_{k_0}\rho_{k_0+1}\al_{k_0+2}^*\wti\rho_{k_0+1} +
\wti\rho_{k_0}\al_{k_0-1}\rho_{k_0}\al_{k_0+1}^*, & k_0 \text{ odd,}
\\
\rho_{k_0+1}\al_{k_0+2}^*\wti\rho_{k_0+1}\al_{k_0} +
\al_{k_0+1}^*\wti\rho_{k_0}\al_{k_0-1}\rho_{k_0}, & k_0 \text{
even.}
\end{cases} \no
\\
&= 2i\oint_{\dD} d\mu_0(\zeta)\, \Xi_{1,1}(\zeta,k_0)\,{\ol
\zeta}^2. \lb{2.19a}
\end{align}
\end{theorem}
\begin{proof}
Let $z\in\D$, $k_0\in\bbZ$. Since
\begin{equation}
\f{\zeta+z}{\zeta-z}=1+2\sum_{j=1}^\infty (\ol\zeta z)^j, \quad
\zeta\in\dD, \lb{2.20}
\end{equation}
\eqref{2.7} implies
\begin{equation}
\ln[M_{1,1}(z,k_0)]=2i\sum_{j=1}^\infty \oint_{\dD} d\mu_0(\zeta)\,
\Xi_{1,1}(\zeta,k_0){\ol\zeta}^j z^j, \, \text{ $|z|$ sufficiently
small}. \lb{2.21}
\end{equation}
A comparison of coefficients of $z^j$ in \eqref{2.16} and
\eqref{2.21} then proves \eqref{2.18}. \eqref{2.19} and
\eqref{2.19a} follow upon substitution of \eqref{2.11} and
\eqref{2.12} into \eqref{2.17} and \eqref{2.17a}.
\end{proof}

\section{Reflectionless Verblunsky Coefficients} \label{s3}

In this section we review basic facts about the half-lattice
Weyl--Titchmarsh $m$-functions and introduce a variety of conditions
for the Verblunsky coefficients $\alpha$ (resp., the CMV operator
$\U$) to be reflectionless. We freely use the notation established
in Appendix \ref{sA}.

Following the presentation of Section 2 in \cite{CGZ07}, we first
recall the four fundamental $m\times m$ matrix valued sequences of
Laurent polynomials $\{P_+(z,k,k_0)$, $Q_+(z,k,k_0)$,
$R_+(z,k,k_0)$, $S_+(z,k,k_0)\}_{k\in\Z}$ associated with the CMV
operator $\U$. These sequences are uniquely defined by the following
difference relations
\begin{align}
&(\W P_+(z,\cdot,k_0))(k) = z R_+(z,k,k_0), \quad (\V
R_+(z,\cdot,k_0))(k) = P_+(z,k,k_0), \no
\\
&(\W Q_+(z,\cdot,k_0))(k) = z S_+(z,k,k_0), \quad (\V
S_+(z,\cdot,k_0))(k) = Q_+(z,k,k_0), \quad k\in\Z, \lb{3.1}
\end{align}
and initial conditions at some reference point $k_0\in\Z$,
\begin{align}
\binom{P_+(z,k_0,k_0)}{R_+(z,k_0,k_0)} &=
\begin{cases}
\binom{zI_m}{I_m}, & \text{$k_0$ odd,} \\[1mm]
\binom{I_m}{I_m}, & \text{$k_0$ even,}
\end{cases} \quad
\binom{Q_+(z,k_0,k_0)}{S_+(z,k_0,k_0)} =
\begin{cases}
\binom{zI_m}{-I_m}, & \text{$k_0$ odd,} \\[1mm]
\binom{-I_m}{I_m}, & \text{$k_0$ even.} \lb{3.2}
\end{cases}
\end{align}
It follows that there exist unique $\Cm$-valued half-lattice
Weyl--Titchmarsh $m$-functions $M_\pm(\cdot,k_0)$ such that for all
$z\in\C\backslash(\dD\cup\{0\})$ the following $\Cm$-valued
sequences have square summable matrix entries, that is,
\begin{align}
\begin{split}
&U_\pm(z,\cdot,k_0) = Q_+(z,\cdot,k_0) +
P_+(z,\cdot,k_0)M_\pm(z,k_0) \in \ltmm{[k_0,\pm\infty)\cap\Z},
\\
&V_\pm(z,\cdot,k_0) = S_+(z,\cdot,k_0) +
R_+(z,\cdot,k_0)M_\pm(z,k_0) \in \ltmm{[k_0,\pm\infty)\cap\Z}.
\end{split} \lb{3.3}
\end{align}
Moreover, one verifies that the functions $M_\pm(\cdot,k_0)|_\D$ are
Caratheodory and anti-Caratheodory matrices, respectively, and hence
extend to the point $z=0$ by analyticity. In addition, the functions $M_\pm$
are intimately related to the half-lattice CMV operators. We refer
to Section 2 in \cite{CGZ07} for a comprehensive study of these
relations.

We will call $U_\pm(z,\cdot,k_0)$ and $V_\pm(z,\cdot,k_0)$ the
Weyl--Titchmarsh solutions associated with $\U$. It follows that
$U_\pm(z,\cdot,k_0)$ and $V_\pm(z,\cdot,k_0)$ are the unique (up to
right-multiplication by constant $m\times m$ matrices)
$\,\Cm$-valued sequences that satisfy difference equations of
the form \eqref{3.1} whose matrix entries are square
summable near $\pm\infty$ (cf.\ \eqref{3.3}).

In applications it is often simpler to manipulate with Schur
matrices rather than Caratheodory ones. To exploit this observation
in the remainder of this section, we introduce (anti)-Schur matrices
$\Phi_\pm(\cdot,k_0)$ associated with (anti)-Caratheodory matrices
$M_\pm(\cdot,k_0)$ by
\begin{align}
\begin{split}
\Phi_\pm(z,k_0) &= [M_\pm(z,k_0)-I_m][M_\pm(z,k_0)+I_m]^{-1}, \quad
\\
M_\pm(z,k_0) &= [I_m+\Phi_\pm(z,k_0)][I_m-\Phi_\pm(z,k_0)]^{-1},
\quad z\in\D.
\end{split} \lb{3.4}
\end{align}
Strictly speaking, one should always consider $\Phi_-^{-1}$ rather
than $\Phi_-$ as $M_-$ is an anti-Caratheodory matrix and hence for
$z\in\D$ the expression $[M_-(z,k_0) + I_m]$ is not necessarily
invertible but $[M_-(z,k_0)-I_m]$ always is (cf.\ \cite[p.\
137]{SF70}). Thus, we should have introduced the Schur matrix
\begin{equation}
\Phi_-(z,k_0)^{-1} = [M_-(z,k_0) + I_m] [M_-(z,k_0) - I_m]^{-1},
\quad z\in\D,
\end{equation}
rather than the anti-Schur matrix $\Phi_-$, but for simplicity of
notation, we will typically avoid this complication with $\Phi_-$
and still invoke $\Phi_-$ rather than $\Phi_-^{-1}$ whenever
confusions are unlikely.

Still following \cite{CGZ07}, we also mention the following two
useful identities that relate the functions $\Phi_\pm$, $\Phi_{1,1}$,
and the Weyl--Titchmarsh solutions $U_\pm$, $V_\pm$ to one another
\begin{align}
&\Phi_\pm(z,k) =
\begin{cases}
zV_\pm(z,k,k_0)U_\pm(z,k,k_0)^{-1}, &\text{$k$ odd,}
\\
U_\pm(z,k,k_0)V_\pm(z,k,k_0)^{-1}, & \text{$k$ even,}
\end{cases}
\quad z\in\D,\; k,k_0\in\bbZ,  \lb{3.6}
\intertext{and}%
&\Phi_{1,1}(z,k_0) =
\begin{cases}
\Phi_-(z,k_0)^{-1}\Phi_+(z,k_0), & k_0\text{ odd},
\\
\Phi_+(z,k_0)\Phi_-(z,k_0)^{-1}, & k_0\text{ even},
\end{cases}
\quad z\in\D,\; k_0\in\bbZ. \lb{3.7}
\end{align}
In addition, the functions $\Phi_\pm(\cdot,k)^{\pm1}$ satisfy the
following Riccati-type equations
\begin{align}
\Phi_+(z,k)\wti\rho_k^{-1}\al_k\Phi_+(z,k-1) +
z\Phi_+(z,k)\wti\rho_k^{-1} - \rho_k^{-1}\Phi_+(z,k-1)
=z\rho_k^{-1}\al_k^*,& \lb{3.8}
\\
z\Phi_-(z,k)^{-1}\rho_k^{-1}\al_k^*\Phi_-(z,k-1)^{-1} +
\Phi_-(z,k)^{-1}\rho_k^{-1} - z\wti\rho_k^{-1}\Phi_-(z,k-1)^{-1}&
\no
\\
= \wti\rho_k^{-1}\al_k, \quad z\in\bbC\backslash\dD,\; k\in\Z.&
\lb{3.8a}
\end{align}

Next, we denote by $M_\pm(\ze,k_0)$, $M_{1,1}(\ze,k_0)$,
$\Phi_\pm(\ze,k_0)$, and $\Phi_{1,1}(\ze,k_0)$, $\ze\in\dD$, etc.,
the radial limits to the unit circle of the corresponding functions,
\begin{align}
M_\pm(\ze,k_0) &= \lim_{r \uparrow 1} M_\pm(r\ze,k_0), &&
M_{1,1}(\ze,k_0) = \lim_{r \uparrow 1} M_{1,1}(r\ze,k_0),
\\
\Phi_\pm(\ze,k_0) &= \lim_{r \uparrow 1} \Phi_\pm(r\ze,k_0), &&
\Phi_{1,1}(\ze,k_0) = \lim_{r \uparrow 1} \Phi_{1,1}(r\ze,k_0),
\quad \ze\in\dD,\; k_0\in\Z. \no
\end{align}
These limits are known to exist Lebesgue almost everywhere on $\dD$.

The following definition of reflectionless Verblunsky coefficients
represents the analog of reflectionless coefficients in
Schr\"odinger, Dirac, Jacobi, and CMV operators (cf., e.g.
\cite{BGMS05}, \cite{CG02}, \cite{CGHL00}, \cite{CGR05} \cite{KS88},
in the matrix-valued coefficients context and \cite{Cr89},
\cite{DS83}, \cite{GKT96}, \cite{GS96}, \cite{GZ06}, \cite{GJ84},
\cite{GJ86}, \cite{Jo82}, \cite{Ko84}--\cite{KS88},
\cite{SY95}--\cite{Te00} in the scalar-valued coefficients context).

\begin{definition} \lb{d3.1}
Assume Hypothesis \ref{h2.1} and let $\U$ be the associated unitary
CMV operator on $\ltm{\Z}$ as defined in \eqref{2.2b}--\eqref{2.3}.
Then $\alpha$ (resp., $\U$) is called {\it reflectionless}, if
\begin{equation}
\text{for some $k_0\in\bbZ$, } \, M_+(\zeta,k_0)=-M_-(\zeta,k_0)^*
\, \text{ for $\mu_0$-a.e.\ $\zeta\in \sigma_{\ess}(\U)$.} \lb{3.9}
\end{equation}
\end{definition}

The following result provides a variety of necessary and sufficient
conditions for $\alpha$ (resp., $\U$) to be reflectionless.

\begin{theorem} \lb{t3.2}
Let $\alpha=\{\alpha_k\}_{k\in\bbZ}$ satisfy Hypothesis \ref{h2.1}
and $\U$ denote the associated unitary CMV operator on $\ltm{\Z}$.
Then the following assertions $(i)$--$(vi)$ are equivalent:
\\
$(i)$ $\alpha=\{\alpha_k\}_{k\in\bbZ}$ is reflectionless.
\\
$(ii)$ $\beta=\{\ga_1\alpha_k\ga_2^*\}_{k\in\bbZ}$ is
reflectionless, where $\ga_1$, $\ga_2$ are $m\times m$ unitary
matrices.
\\
$(iii)$ For some $k_0\in\bbZ$, $M_+(\zeta,k_0)^*=-M_-(\zeta,k_0)$
for $\mu_0$-a.e.\ $\zeta\in \sigma_{\ess}(\U)$.
\\
$(iv)$ For all $k\in\bbZ$, $M_+(\zeta,k)^*=-M_-(\zeta,k)$ for
$\mu_0$-a.e.\ $\zeta\in \sigma_{\ess}(\U)$.
\\
$(v)$ For some $k_0\in\bbZ$,
$\Phi_+(\zeta,k_0)^*=\Phi_-(\zeta,k_0)^{-1}$ for $\mu_0$-a.e.\
$\zeta\in\sigma_{\ess}(\U)$.
\\
$(vi)$ For all $k\in\bbZ$, $\Phi_+(\zeta,k)^*=\Phi_-(\zeta,k)^{-1}$
for $\mu_0$-a.e.\ $\zeta\in\sigma_{\ess}(\U)$.
\\[1mm]
Moreover, conditions $(i)$--$(vi)$ imply the following equivalent
assertions $(vii)$--$(ix)$:
\\
$(vii)$ For all $k\in\bbZ$, $M_{1,1}(\zeta,k) > 0$ for $\mu_0$-a.e.\
$\zeta\in\sigma_{\ess}(\U)$.
\\
$(viii)$ For all $k\in\bbZ$, $-I_m<\Phi_{1,1}(\zeta,k)<I_m$ for
$\mu_0$-a.e.\ $\zeta\in\sigma_{\ess}(\U)$.
\\
$(ix)$ For all $k\in\bbZ$, $\Xi_{1,1}(\zeta,k)=0$ for $\mu_0$-a.e.\
$\zeta\in\sigma_{\ess}(\U)$.
\\
\end{theorem}
\begin{proof}
We start by noting that $(i)$ is equivalent to $(iii)$ by Definition
\eqref{d3.1}, $(iii)$ and $(iv)$ are equivalent to $(v)$ and $(vi)$,
respectively, by \eqref{3.4}, $(vii)$ is equivalent to $(viii)$ by
\eqref{2.3h}, and $(vii)$ is equivalent to $(ix)$ by \eqref{2.8}.

Next, we show that $(v)$ is equivalent to $(vi)$. One direction is
trivial and for the other it suffices to check that $(v)$ at some
point $k_0$ implies $(v)$ at points $k_0\pm1$. Taking adjoint in
\eqref{3.8} and solving for $\Phi_+(\ze,k_0-1)^*$ one obtains
\begin{align}
\Phi_+(\ze,k_0-1)^*=
\big(\ol{\ze}\wti\rho_{k_0}^{-1}\Phi_+(\ze,k_0)^* -
\ol{\ze}\al_{k_0}\rho_{k_0}^{-1}\big)
\big(-\al_{k_0}^*\wti\rho_{k_0}^{-1}\Phi_+(\ze,k_0)^* +
\rho_{k_0}^{-1}\big)^{-1}. \lb{3.9a}
\end{align}
Similarly, solving \eqref{3.8a} for $\Phi_-(\ze,k_0)^{-1}$ one
computes
\begin{align}
\Phi_-(\ze,k_0)^{-1} &=
\big(\ze\wti\rho_{k_0}^{-1}\Phi_-(\ze,k_0-1)^{-1} +
\wti\rho_{k_0}^{-1}\al_{k_0}\big) \no
\\
&\quad\times
\big(\ze\rho_{k_0}^{-1}\al_{k_0}^*\Phi_-(\ze,k_0-1)^{-1} +
\rho_{k_0}^{-1}\big)^{-1}. \lb{3.9b}
\end{align}
Since by $(v)$ $\Phi_+(\ze,k_0)^* = \Phi_-(\ze,k_0)^{-1}$ for
$\mu_0$-a.e.\ $\zeta\in\sigma_{\ess}(\U)$, insertion of \eqref{3.9b}
into \eqref{3.9a} yields
\begin{align}
\Phi_+(\ze,k_0-1)^* &=
\big(\wti\rho_{k_0}^{-2}\Phi_-(\ze,k_0-1)^{-1} +
\ol{\ze}\wti\rho_{k_0}^{-2}\al_{k_0} \no
\\
&\quad\quad -
\al_{k_0}\rho_{k_0}^{-2}\al_{k_0}^*\Phi_-(\ze,k_0-1)^{-1} -
\ol{\ze}\al_{k_0}\rho_{k_0}^{-2} \big) \no
\\
&\quad\times
\big(-\ze\al_{k_0}^*\wti\rho_{k_0}^{-2}\Phi_-(\ze,k_0-1)^{-1} -
\al_{k_0}^*\wti\rho_{k_0}^{-2}\al_{k_0}^* \no
\\
&\quad\quad + \ze\rho_{k_0}^{-2}\al_{k_0}^*\Phi_-(\ze,k_0-1)^{-1} +
\rho_{k_0}^{-2}\big)^{-1} \no
\\
& = \Phi_-(\ze,k_0-1)^{-1} \quad \text{for $\mu_0$-a.e.\
$\zeta\in\sigma_{\ess}(\U)$}.
\end{align}
Here \eqref{2.2} and \eqref{2.2a} were used to simplify the
expression. Thus, $(v)$ at $k_0$ implies $(v)$ at $k_0-1$.
Similarly, one shows that $(v)$ at $k_0$ also implies $(v)$ at
$k_0+1$.

Next, we verify that $(iii)$ implies $(vii)$. Recall that by Lemma
3.3 in \cite{CGZ07} the resolvent $(\U-zI)^{-1}$ is given in terms
of its matrix elements in the standard basis of $\ltm{\Z}$ by
\begin{align}
(\U-zI)^{-1}(k,k') = \frac{1}{2z}
\begin{cases}
U_-(z,k,k_0)W(z,k_0)^{-1}U_+(1/\ol{z},k',k_0)^*,
\\\hspace{27mm} k < k' \text{ or } k = k' \text{ odd},
\\
U_+(z,k,k_0)W(z,k_0)^{-1}U_-(1/\ol{z},k',k_0)^*,
\\\hspace{26mm} k > k' \text{ or } k = k' \text{ even},
\end{cases} \lb{3.10}
\end{align}
where $W(z,k_0)=M_+(z,k_0)-M_-(z,k_0)$ is the Wronskian of $U_+$ and
$U_-$. Using $(iii)$, \eqref{3.3}, \eqref{A.7}, and the fact that
$P_+(z,k,k_0)$ and $Q_+(z,k,k_0)$ are Laurent polynomials in $z$ and
hence are analytic in $\C\bs\{0\}$, one computes
\begin{align}
\lim_{r \uparrow 1}U_\pm(1/(r\ol{\ze}),k,k_0) &= \lim_{r \uparrow
1}[Q_+(\ze/r,k,k_0)-P_+(\ze/r,k,k_0)M_\pm(r\ze,k_0)^*] \no
\\
&=Q_+(\ze,k,k_0)-P_+(\ze,k,k_0)M_\pm(\ze,k_0)^* \no
\\
&=Q_+(\ze,k,k_0)+P_+(\ze,k,k_0)M_\mp(\ze,k_0) \no
\\
&=\lim_{r \uparrow 1}U_\mp(r\ze,k,k_0) \,\text{ for $\mu_0$-a.e.\
$\zeta\in\sigma_{\ess}(\U)$}. \lb{3.11}
\end{align}
Thus, combining $(iii)$, \eqref{2.3d}, \eqref{3.10}, and
\eqref{3.11} one concludes for all $k\in\Z$ and $\mu_0$-a.e.\
$\zeta\in\sigma_{\ess}(\U)$ that
$W(\ze,k_0)=M_+(\ze,k_0)+M_+(\ze,k_0)^*=W(\ze,k_0)^*$, and hence,
\begin{align}
M_{1,1}(\ze,k) &= I_m + 2z\De_{k}(\U-zI)^{-1}\De_{k} = I_m +
2z(\U-zI)^{-1}(k,k) \no
\\
&=I_m +
\begin{cases}
U_-(\ze,k,k_0)W(z,k_0)^{-1}U_-(\ze,k,k_0)^*, & k \text{ odd},
\\
U_+(\ze,k,k_0)W(z,k_0)^{-1}U_+(\ze,k,k_0)^*, & k \text{ even}
\end{cases} \lb{3.12}
\end{align}
is a nonnegative Caratheodory matrix, that is, $(vii)$ holds.

Finally, we check that $(i)$ is equivalent to $(ii)$. First, note
that $\be_k=\ga_1\al_k\ga_2^*$, $k\in\Z$, implies the following
relations for matrices defined in \eqref{2.2}--\eqref{2.3}
associated with Verblunsky coefficients $\al$ and $\be$,
respectively,
\begin{align}
&\rho_{\be;k} = \ga_2\rho_{\al;k}\ga_2^*, \quad \wti\rho_{\be;k} =
\ga_1\wti\rho_{\al;k}\ga_1^*,
\\
&\Te_{\be;k} =
\begin{pmatrix}\ga_1&0\\0&\ga_2\end{pmatrix}
\Te_{\al;k}
\begin{pmatrix}\ga_2&0\\0&\ga_1\end{pmatrix}^*, \quad k\in\Z,
\end{align}
and hence,
\begin{align}
& \V_{\be}=\Ga_1\V_{\al}\Ga_2^*, \quad
\W_{\be}=\Ga_2\W_{\al}\Ga_1^*, \quad \U_{\be}=\Ga_1\U_{\al}\Ga_1^*,
\lb{3.16}
\end{align}
where $\Ga_1$ and $\Ga_2$ are block-diagonal unitary operators on
$\ltm{\Z}$ with diagonals given by
\begin{align}
\Ga_1(k,k) =
\begin{cases}
\ga_1, & k \text{ odd},
\\
\ga_2, & k \text{ even},
\end{cases} \quad
\Ga_2(k,k) =
\begin{cases}
\ga_2, & k \text{ odd},
\\
\ga_1, & k \text{ even},
\end{cases} \quad k\in\Z. \lb{3.17}
\end{align}
Then it follows from \eqref{3.16} and the definition of the
Weyl--Titchmarsh solutions $U_{\be;\pm}$, $V_{\be;\pm}$ associated
with the CMV operator $\U_\be$,
\begin{align}
&\W_{\be} U_{\be;\pm}(z,\cdot,k_0) = z V_{\be;\pm}(z,\cdot,k_0),
\quad \V_{\be} V_{\be;\pm}(z,\cdot,k_0) = U_{\be;\pm}(z,\cdot,k_0),
\no
\\
&U_{\be;\pm}(z,\cdot,k_0), \, V_{\be;\pm}(z,\cdot,k_0) \in
\ltmm{[k_0,\pm\infty)\cap\Z},
\end{align}
that the $\Cm$-valued sequences $\Ga_1^*U_{\be;\pm}$ and
$\Ga_2^*V_{\be;\pm}$ satisfy
\begin{align}
&\W_{\al} \Ga_1^*U_{\be;\pm}(z,\cdot,k_0) = z
\Ga_2^*V_{\be;\pm}(z,\cdot,k_0), \quad \V_{\al}
\Ga_2^*V_{\be;\pm}(z,\cdot,k_0) = \Ga_1^*U_{\be;\pm}(z,\cdot,k_0),
\no
\\
&\Ga_1^* U_{\be;\pm}(z,\cdot,k_0), \, \Ga_2^*
V_{\be;\pm}(z,\cdot,k_0) \in \ltmm{[k_0,\pm\infty)\cap\Z}.
\end{align}
Thus, the uniqueness of the Weyl--Titchmarsh solutions associated
with $\U_{\al}$ implies that there is an $m\times m$ matrix $C$ such
that
\begin{align}
\Ga_1^*U_{\be;\pm}(z,k,k_0)=U_{\al;\pm}(z,k,k_0)C, \quad
\Ga_2^*V_{\be;\pm}(z,k,k_0)=V_{\al;\pm}(z,k,k_0)C,
\end{align}
equivalently,
\begin{align}
& U_{\al;\pm}(z,k,k_0) =
\begin{cases}
\ga_1^*U_{\be;\pm}(z,k,k_0)C, & k \text{ odd},\\
\ga_2^*U_{\be;\pm}(z,k,k_0)C, & k \text{ even},
\end{cases} \lb{3.21}
\\
& V_{\al;\pm}(z,k,k_0)=
\begin{cases}
\ga_2^*V_{\be;\pm}(z,k,k_0)C, & k \text{ odd},\\
\ga_1^*V_{\be;\pm}(z,k,k_0)C, & k \text{ even},
\end{cases} \quad k\in\Z. \lb{3.22}
\end{align}
Inserting \eqref{3.21} and \eqref{3.22} into \eqref{3.6} we get
\begin{align}
\Phi_{\al;\pm}(z,k) = \ga_2^*\Phi_{\be;\pm}(z,k)\ga_1, \quad z\in\D,
\; k\in\Z. \lb{3.23}
\end{align}
Since by \eqref{3.16} $\U_{\al}$ is unitarily equivalent to
$\U_{\be}$ and hence $\si_\ess(\U_{\al})=\si_\ess(\U_{\be})$, we
conclude from \eqref{3.23} that $\Phi_{\al;\pm}(z,k)$ satisfy $(v)$
if and only if $\Phi_{\be;\pm}(z,k)$ do. Hence, the previously
established equivalence of $(i)$ and $(v)$ finishes the proof.
\end{proof}

It is instructive to state as a separate result the following fact
obtained in the proof of Theorem \ref{t3.2} (cf. \eqref{3.16} and
\eqref{3.17}).

\begin{theorem} \lb{t3.3}
Let $\al=\{\al_k\}_{k\in\bbZ}$ be a sequence satisfying Hypothesis
\ref{h2.1} and fix two $m\times m$ unitary matrices $\gamma_1$,
$\gamma_2$. Define $\be=\{\gamma_1\alpha_k\gamma_2^*\}_{k\in\bbZ}$.
Then the CMV operators $U_\al$ and $U_\be$ associated with $\al$ and
$\be$, respectively, are unitarily equivalent.
\end{theorem}

\section{The Borg-Type Theorem for CMV Operators} \label{s4}

In this section we finally prove our principal new result, a general
Borg-type theorem for reflectionless CMV operators with spectrum a
connected subarc of the unit circle. We freely use the notation
established in Appendix \ref{sA}.

First, we prove the following uniqueness result which is a special
case of Borg-type theorem for reflectionless CMV operators.
\begin{theorem}  \lb{t4.1}
Let $\alpha=\{\alpha_k\}_{k\in\bbZ}$ be a reflectionless sequence of
$m\times m$ matrix-valued Verblunsky coefficients. Let $\U$ be the
associated unitary CMV operator \eqref{2.3} on $\ltm{\Z}$ and
suppose that
\begin{equation}
\sigma(\U)=\dD. \lb{4.1}
\end{equation}
Then $\alpha=\{\alpha_k\}_{k\in\bbZ}$ is of the form,
\begin{equation}
\alpha_k=0, \quad k\in\bbZ. \lb{4.2}
\end{equation}
\end{theorem}
\begin{proof}
Since by hypothesis $\U$ is reflectionless, one infers from
Definition \ref{d3.1} and Theorem \ref{t3.2} that
\begin{equation}
\Phi_+(\ze,k)^*=\Phi_-(\ze,k)^{-1}, \quad \mu_0\text{-a.e. }
\ze\in\dD, \; k\in\Z \lb{4.3}
\end{equation}
and
\begin{equation}
\Xi_{1,1}(\ze,k)=0, \quad \mu_0\text{-a.e. } \ze\in\dD, \; k\in\Z.
\lb{4.4}
\end{equation}
Then it follows from \eqref{2.7} and \eqref{4.4} that
$M_{1,1}(z,k)=I_m$ for all $z\in\D$, $k\in\Z$, and hence by
\eqref{2.3h}
\begin{align}
\Phi_{1,1}(z,k)=0, \quad z\in\D, \; k\in\Z. \lb{4.5}
\end{align}
This together with \eqref{3.7} and \eqref{4.3} implies
\begin{align}
\Phi_+(\ze,k)=\Phi_-(\ze,k)^{-1}=0, \quad \mu_0\text{-a.e. }
\ze\in\dD, \; k\in\Z. \lb{4.6}
\end{align}
Taking radial limits in the Riccati-type equations \eqref{3.8},
\eqref{3.8a} and substituting \eqref{4.6} into the left hand-sides
yield
\begin{align}
\al_{k}=0, \quad k\in\Z. \lb{4.9}
\end{align}
\end{proof}

Next, we introduce the following notation for closed arcs on the
unit circle $\dD$,
\begin{equation}
\Arc\big(\big[e^{i\theta_1},e^{i\theta_2}\big]\big)
=\big\{e^{i\theta}\in\dD\,|\, \theta_1\leq\theta\leq \theta_2\big\},
\quad \theta_1 \in [0,2\pi), \; \theta_1\leq \theta_2\leq
\theta_1+2\pi
\end{equation}
and similarly for open arcs and arcs open or closed at one endpoint
(cf.\ \eqref{A.5}).

The principal new result of this paper then reads as follows.

\begin{theorem}  \lb{t5.3}
Let $\alpha=\{\alpha_k\}_{k\in\bbZ}$ be a reflectionless sequence of
$m\times m$ matrix-valued Verblunsky coefficients. Let $\U$ be the
associated unitary CMV operator \eqref{2.3} on $\ltm{\Z}$ and
suppose that the spectrum of $\U$ consists of a connected arc of
$\dD$,
\begin{equation}
\sigma(\U)=\Arc\big(\big[e^{i\theta_0},e^{i\theta_1}\big]\big)
\lb{4.12}
\end{equation}
with $\theta_0 \in [0,2\pi)$, $\theta_0<\theta_1\leq\theta_0+2\pi$.
Then $\alpha=\{\alpha_k\}_{k\in\bbZ}$ is of the form,
\begin{equation}
\alpha_k=g^k a \ga, \quad k\in\bbZ, \lb{4.13}
\end{equation}
where
\begin{equation}
g=-\exp(i(\theta_0+\theta_1)/2), \quad
a=\cos((\theta_1-\theta_0)/4), \lb{4.14}
\end{equation}
and $\ga$ is some $k$-independent $m\times m$ unitary matrix.
\end{theorem}
\begin{proof}
First, note that in the special case $\si(\U)=\dD$ the result
follows from Theorem \ref{t4.1}. Hence without loss of generality we
will assume in the following that $\si(\U)\subsetneq\dD$, that is,
$\te_1-\te_0<2\pi$. Next, we proceed with the proof in two steps.

In our first step we find an explicit formula for the function
$\Xi_{1,1}(\cdot,k)$ on $\dD$, $k\in\Z$. To understand the behavior
of $\Xi_{1,1}(\cdot,k)$ it suffices by \eqref{2.8} to study the
behavior of the boundary values of $M_{1,1}(\cdot,k)$ on $\dD$.

We start by noting that Theorem \ref{tA.2} (cf. \eqref{A.4}) implies
that the Caratheodory matrix $M_{1,1}(\cdot,k)$ has purely imaginary
(i.e., $\Re(M_{1,1}(\cdot,k))=0$) boundary values $\mu_0\text{-a.e.}$
on $\dD\bs\supp(d\Om_{1,1}(\cdot,k))$. Moreover, it follows from
\eqref{2.3b} and Theorem \ref{t2.2} that
\begin{align}
\supp(d\Om_{1,1}(\cdot,k))\subseteq\supp(d\Om(\cdot,k))=\si(\U),
\quad k\in\Z. \lb{4.15}
\end{align}
On the other hand, the reflectionless assumption, Theorem \ref{t3.2}
$(vii)$, and \eqref{4.12} imply that $M_{1,1}(\cdot,k)$ has strictly
positive (i.e., $M_{1,1}(\cdot,k)>0$) boundary values
$\mu_0\text{-a.e.}$ on $\si_\ess(\U)=\si(\U)$. By Theorem \ref{tA.2}
(cf. \eqref{A.4}) this implies that
\begin{align}
\si(\U)=\si_\ess(\U)\subseteq\supp(d\Om_{1,1}(\cdot,k)), \quad
k\in\Z. \lb{4.16}
\end{align}
Thus, it follows from \eqref{4.15} and \eqref{4.16} that
$\si(\U)=\supp(d\Om_{1,1}(\cdot,k))$, $k\in\Z$. The same argument
actually implies more, namely, that for any $x_0\in\C^m$ with
$\|x_0\|_{\C^m}=1$, the scalar-valued Caratheodory function
$m_{x_0}(\cdot,k)$ defined by
\begin{align}
& m_{x_0}(z,k) = \big(x_0,M_{1,1}(z,k)x_0\big)_{\C^m} = \oint_{\dD}
d\om_{x_0}(\zeta,k) \, \f{\zeta+z}{\zeta-z}, \lb{4.17}
\\
& d\om_{x_0}(\cdot,k) =
d\big(x_0,\Omega_{1,1}(\cdot,k)x_0\big)_{\C^m}, \quad z\in\D, \;
k\in\Z,
\end{align}
is purely imaginary $\mu_0\text{-a.e.}$ on $\dD\bs\si(\U)$, strictly
positive $\mu_0\text{-a.e.}$ on $\si(\U)$, and
$\supp(d\om_{x_0}(\cdot,k))=\si(\U)$ for all $k\in\Z$.
Differentiating $-im_{x_0}(e^{i\te},k)$ with respect to $\te$ shows
that $\Im(m_{x_0}(\cdot,k))=-im_{x_0}(\cdot,k)$ is monotone
decreasing on $\dD\bs\si(\U)$,
\begin{equation}
\f{d}{d\theta}\big(-im_{x_0}(e^{i\theta},k)\big)=-\f{1}{2}
\int_{[0,2\pi)\backslash[\theta_0,\theta_1]}
\f{d\om_{x_0}(e^{it},k)}{\sin^2((t-\theta)/2)}<0, \quad
\theta\in(\theta_0,\theta_1). \lb{4.19}
\end{equation}
This implies that there exists a
$\theta_*(x_0,k)\in[\te_1,\te_0+2\pi]$ such that the exponential
Herglotz representation for $m_{x_0}(\cdot,k)$,
\begin{align}
& \ln[m_{x_0}(z,k)]=i \oint_{\dD} d\mu_0(\zeta)\, \xi_{k_0}(\zeta,k)
\f{\zeta+z}{\zeta-z}, \quad z\in\D, \lb{4.20}
\end{align}
yields the following form for the function $\xi_{x_0}(\zeta,k)=\lim_{r\uparrow 1}
\Im[\ln(m_{x_0}(r\zeta,k))]$,
\begin{align}
\xi_{x_0}(\zeta,k)=
\begin{cases} 0, &
\zeta\in\Arc\big(\big(e^{i\theta_0},e^{i\theta_1}\big)\big),\\
\pi/2, &
\zeta\in\Arc\big(\big(e^{i\theta_1},e^{i\theta_*(x_0,k)}\big)\big),\\
-\pi/2, &
\zeta\in\Arc\big(\big(e^{i\theta_*(x_0,k)},e^{i(\theta_0+2\pi)}\big)\big)
\end{cases}
\,\text{ for $\mu_0$-a.e.\ $\zeta\in\dD$}. \lb{4.21}
\end{align}
Since by \eqref{2.3f} $m_{x_0}(0,k)=\|x_0\|_{\C^m}^2=1$, $k\in\Z$,
we compute using \eqref{4.20} and \eqref{4.21}
\begin{align}
0&=\ln[m_{x_0}(0,k)]=\oint_{\dD}d\mu_0(\zeta)\,\Xi_{1,1}(\zeta,k) =
\f{1}{4} \oint_{\theta_1}^{\theta_*(k,x_0)} d\theta - \f{1}{4}
\oint_{\theta_*(x_0,k)}^{\theta_0+2\pi} d\theta \no \\
&=\f{1}{4}[2\theta_*(x_0,k)-\theta_0-\theta_1-2\pi], \quad k\in\Z,
\lb{4.22}
\end{align}
and hence
\begin{equation}
\theta_*(x_0,k)=\f{1}{2}(\theta_0+\theta_1)+\pi, \quad k\in\Z,
\lb{4.23}
\end{equation}
is in fact $(x_0,k)$-independent and denoted by $\theta_*$ in the
following. As a result, $\xi_{x_0}(\cdot,k)=\xi(\cdot)$ in
\eqref{4.21} and hence $m_{x_0}(\cdot,k)=m(\cdot)$ in \eqref{4.20}
are also $(x_0,k)$-independent. Recalling \eqref{2.8}, \eqref{4.17},
and \eqref{4.21} we conclude that $M_{1,1}(\cdot,k)=M_{1,1}(\cdot)$
and hence $\Xi_{1,1}(\cdot,k)=\Xi_{1,1}(\cdot)$ are $k$-independent
and
\begin{equation}
\Xi_{1,1}(\zeta)=\begin{cases} 0, &
\zeta\in\Arc\big(\big(e^{i\theta_0},e^{i\theta_1}\big)\big), \\
\f{\pi}{2}I_m, &
\zeta\in\Arc\big(\big(e^{i\theta_1},e^{i\theta_*}\big)\big), \\
-\f{\pi}{2}I_m, &
\zeta\in\Arc\big(\big(e^{i\theta_*},e^{i(\theta_0+2\pi)}\big)\big)
\end{cases}
\,\text{ for $\mu_0$-a.e.\ $\zeta\in\dD$}. \lb{4.24}
\end{equation}

In our second step we use the above explicit form of the function
$\Xi_{1,1}$ and the trace formulas obtained in Theorem \ref{t2.4} to
derive various identities for Verblunsky coefficients $\al$ which
will imply \eqref{4.13} and \eqref{4.14}.

By \eqref{4.24} the following matrix
\begin{align}
-i\oint_{\dD} d\mu_0(\zeta)\, \Xi_{1,1}(\zeta)\,{\ol \zeta} &= -i
\oint_{\te_1}^{\te_*}\frac{\pi}{2}e^{-it}\frac{dt}{2\pi} I_m +
i\oint_{\te_*}^{\te_0+2\pi}\frac{\pi}{2}e^{-it}\frac{dt}{2\pi} I_m
\no \\
&= -\frac{1}{4}e^{-i(\te_0+\te_1)/2}
\big(2+2\cos((\te_1-\te_0)/2)\big)I_m \no
\\ &=
-e^{-i(\te_0+\te_1)/2} \cos^2((\te_1-\te_0)/4)I_m
\end{align}
is a nonzero scalar multiple of the identity matrix $I_m$ since
$0<\te_1-\te_0<2\pi$. Hence, it follows from \eqref{2.19} that
$\al_k$ is nonsingular, commutes with $\al_{k+1}^*$, and
\begin{align}
\al_{k+1}^*\al_{k} = \al_{k}\al_{k+1}^* = -e^{-i(\te_0+\te_1)/2}
\cos^2((\te_1-\te_0)/4)I_m, \quad k\in\Z. \label{4.26}
\end{align}
Combining \eqref{2.2a} with \eqref{4.26} one also gets
\begin{align}
\al_{k+1}^*\wti\rho_k = \rho_k\al_{k+1}^* \,\text{ and }\,
\wti\rho_k\al_{k-1}=\al_{k-1}\rho_k, \quad k\in\Z. \lb{4.27}
\end{align}
Similarly, by \eqref{4.24}
\begin{align}
-i\oint_{\dD} d\mu_0(\zeta)\, \Xi_{1,1}(\zeta)\,{\ol \zeta}^2 &= -i
\oint_{\te_1}^{\te_*}\frac{\pi}{2}e^{-2it}\frac{dt}{2\pi} I_m +
i\oint_{\te_*}^{\te_0+2\pi}\frac{\pi}{2}e^{-2it}\frac{dt}{2\pi} I_m
\lb{4.28}
\no \\
&= \frac{1}{4}e^{-i(\te_0+\te_1)} \big(1-\cos(\te_1-\te_0)\big)I_m
\\
&= 2e^{-i(\te_0+\te_1)}
\big(\cos^2((\te_1-\te_0)/4)-\cos^4((\te_1-\te_0)/4)\big)I_m \no
\end{align}
is also a nonzero scalar multiple of the identity matrix $I_m$.
Hence, \eqref{2.19a} together with \eqref{4.26}--\eqref{4.28}
implies
\begin{align}
&\rho_{k+1}\al_{k+2}^*\wti\rho_{k+1}\al_{k} +
\al_{k+1}^*\wti\rho_{k}\al_{k-1}\rho_{k} =
\al_{k}\rho_{k+1}\al_{k+2}^*\wti\rho_{k+1} +
\wti\rho_{k}\al_{k-1}\rho_{k}\al_{k+1}^* \no
\\
&\quad = \al_{k}\rho_{k+1}^2\al_{k+2}^* +
\al_{k-1}\rho_{k}^2\al_{k+1}^* \no
\\
&\quad = 2e^{-i(\te_0+\te_1)}
\big(\cos^2((\te_1-\te_0)/4)-\cos^4((\te_1-\te_0)/2)\big)I_m, \quad
k\in\Z. \label{4.29}
\end{align}
Inserting $\rho_k^2=I_m-\al_k^*\al_k$ into the last equality of
\eqref{4.29} and simplifying the expression using \eqref{4.26} yield
\begin{align}
\al_{k}\al_{k+2}^* + \al_{k-1}\al_{k+1}^* = 2e^{-i(\te_0+\te_1)}
\cos^2((\te_1-\te_0)/4)I_m, \quad k\in\Z. \lb{4.30}
\end{align}
Multiplying both sides on the left and right by $\al_k^*$ and
$\al_{k+1}$, respectively, and using \eqref{4.26} once again, imply
\begin{align}
\al_{k}^*\al_{k} + \al_{k+1}^*\al_{k+1} =
2\cos^2((\te_1-\te_0)/4)I_m, \quad k\in\Z. \lb{4.31}
\end{align}
Then \eqref{4.26} and \eqref{4.31} imply
\begin{align}
\al_{k+1}=-e^{i(\te_0+\te_1)/2}\al_k, \quad k\in\Z, \lb{4.32}
\end{align}
since
\begin{align}
&\big(\al_{k+1}+e^{i(\te_0+\te_1)/2}\al_k\big)^*
\big(\al_{k+1}+e^{i(\te_0+\te_1)/2}\al_k\big) \no
\\
&\quad = \al_{k}^*\al_{k} + \al_{k+1}^*\al_{k+1} +
e^{i(\te_0+\te_1)/2}\al_{k+1}^*\al_{k} +
e^{-i(\te_0+\te_1)/2}\al_{k}^*\al_{k+1} \no
\\
&\quad = 2\cos^2((\te_1-\te_0)/4)I_m-2\cos^2((\te_1-\te_0)/4)I_m=0,
\quad k\in\Z.
\end{align}
Inserting \eqref{4.32} into \eqref{4.26} yields
\begin{align}
\al_k^*\al_k = \cos^2((\te_1-\te_0)/4)I_m, \quad k\in\Z. \lb{4.34}
\end{align}
Finally, defining the unitary matrix $\ga$ by
$\ga=\al_0(\al_0^*\al_0)^{-1/2}$ one obtains \eqref{4.13} and
\eqref{4.14} from \eqref{4.32} and \eqref{4.34}.
\end{proof}

\begin{remark}  \lb{r5.4}
By Theorems \ref{t3.3} and \ref{t3.2} $(ii)$ the unitary matrix $\ga$
in \eqref{4.13} is a unitary invariant that preserves the
reflectionless property, hence necessarily remains undetermined.
\end{remark}

\appendix
\section{Basic Facts on Caratheodory and Schur Functions}
\lb{sA}
\renewcommand{\theequation}{A.\arabic{equation}}
\renewcommand{\thetheorem}{A.\arabic{theorem}}
\setcounter{theorem}{0} \setcounter{equation}{0}

In this appendix we summarize a few basic facts on matrix-valued
Caratheodory and Schur functions used throughout this manuscript.
(For the analogous case of matrix-valued Herglotz functions we refer
to \cite{GT00} and the extensive list of references therein.)

We denote by $\D$ and $\dD$ the open unit disk and the
counterclockwise oriented unit circle in the complex plane $\C$,
\begin{equation}
\D = \{ z\in\C \st \abs{z} < 1 \}, \quad \dD = \{ \ze\in\C \st
\abs{\ze} = 1 \}.
\end{equation}
Moreover, we denote as usual $\Re(A)=(A+A^*)/2$ and
$\Im(A)=(A-A^*)/(2i)$ for square matrices $A$ with complex-valued
entries.

\begin{definition} \lb{dA.1}
Let $m\in\bbN$ and $F_\pm$, $\Phi_+$, and $\Phi_-^{-1}$ be $m\times
m$
matrix-valued analytic functions in $\D$. \\
$(i)$ $F_+$ is called a {\it Caratheodory matrix} if
$\Re(F_+(z))\geq 0$ for all $z\in\D$ and $F_-$ is called an {\it
anti-Caratheodory matrix} if $-F_-$ is a
Caratheodory matrix. \\
$(ii)$ $\Phi_+$ is called a {\it Schur matrix} if
$\|\Phi_+(z)\|_{\Cm} \leq 1$, for all $z\in\D$.\  $\Phi_-$ is called
an {\it anti-Schur matrix} if $\Phi_-^{-1}$ is a Schur matrix.
\end{definition}

\begin{theorem} \lb{tA.2}
Let $F$ be an $m\times m$ Caratheodory matrix, $m\in\bbN$. Then $F$
admits the Herglotz representation
\begin{align}
& F(z)=iC+ \oint_{\dD} d\Omega(\zeta) \, \f{\zeta+z}{\zeta-z}, \quad
z\in\D, \lb{A.3}
\\
& C=\Im(F(0)), \quad \oint_{\dD} d\Omega(\zeta) = \Re(F(0)),
\end{align}
where $d\Omega$ denotes a nonnegative $m \times m$ matrix-valued
measure on $\dD$. The measure $d\Omega$ can be reconstructed from
$F$ by the formula
\begin{equation}
\Omega\big(\Arc\big(\big(e^{i\te_1},e^{i\te_2}\big]\big)\big)
=\lim_{\delta\downarrow 0} \lim_{r\uparrow 1} \f{1}{2\pi}
\oint_{\te_1+\delta}^{\te_2+\delta} d\te \,
\Re\big(F\big(r\zeta\big)\big), \lb{A.4}
\end{equation}
where
\begin{equation}
\Arc\big(\big(e^{i\theta_1},e^{i\theta_2}\big]\big)
=\big\{\zeta\in\dD\,|\, \theta_1<\te\leq \theta_2\big\}, \quad
\theta_1 \in [0,2\pi), \; \theta_1<\theta_2\leq \theta_1+2\pi.
\lb{A.5}
\end{equation}
Conversely, the right-hand side of equation \eqref{A.3} with $C =
C^*$ and $d\Omega$ a finite nonnegative $m \times m$ matrix-valued
measure on $\dD$ defines a Caratheodory matrix.
\end{theorem}

We note that additive nonnegative $m\times m$ matrices on the
right-hand side of \eqref{A.3} can be absorbed into the measure
$d\Om$ since
\begin{equation}
\oint_\dD d\mu_0(\zeta) \, \f{\zeta+z}{\zeta-z}=1, \quad z\in\D,
\lb{A.5a}
\end{equation}
where
\begin{equation}
d\mu_0(\zeta)=\f{d\te}{2\pi}, \quad \zeta=e^{i\te}, \; \te\in
[0,2\pi),   \lb{A.5b}
\end{equation}
denotes the normalized Lebesgue measure on the unit circle $\dD$.

Given a Caratheodory (resp., anti-Caratheodory) matrix $F_+$ (resp.
$F_-$) defined in $\D$ as in \eqref{A.3}, one extends $F_\pm$ to all
of $\bbC\backslash\dD$ by
\begin{equation}
F_\pm(z)=iC_\pm \pm \oint_{\dD} d\Om_\pm (\zeta) \,
\f{\zeta+z}{\zeta-z}, \quad z\in\bbC\backslash\dD, \;\;
C_\pm=C_\pm^*. \lb{A.6}
\end{equation}
In particular,
\begin{equation}
F_\pm(z) = -F_\pm(1/\ol{z})^*, \quad z\in\C\backslash\ol{\D}.
\lb{A.7}
\end{equation}
Of course, this continuation of $F_\pm|_{\D}$ to
$\bbC\backslash\ol\D$, in general, is not an analytic continuation
of $F_\pm|_\D$.

Next, given the functions $F_\pm$ defined in $\bbC\backslash\dD$ as
in \eqref{A.6}, we introduce the functions $\Phi_\pm$ by
\begin{equation}
\Phi_\pm(z)=[F_\pm(z)-I_m][F_\pm(z)+I_m]^{-1}, \quad
z\in\bbC\backslash\dD.  \lb{A.11}
\end{equation}
We recall (cf., e.g., \cite[p.\ 167]{SF70}) that if $\pm \Re(F_\pm)
\geq 0$, then $[F_\pm \pm I_m]$ is invertible. In particular,
$\Phi_+|_{\D}$ and $[\Phi_-]^{-1}|_{\D}$ are Schur matrices (resp.,
$\Phi_-|_{\D}$ is an anti-Schur matrix). Moreover,
\begin{equation}
F_\pm(z)= [I_m-\Phi_\pm (z)]^{-1} [I_m+\Phi_\pm (z)] \lb{A.12}
\end{equation}
and
\begin{align}
\Phi_\pm(1/\ol{z}) = \big[\Phi_\pm(z)^*\big]^{-1}, \quad
z\in\bbC\backslash\dD. \lb{A.13}
\end{align}
\medskip

{\bf Acknowledgments.} We are indebted to Fritz Gesztesy and Eric
Ryckman for valuable comments and helpful discussions on this topic.


\end{document}